\documentclass{CSML}
\pdfoutput=1

\usepackage{lastpage}
\newcommand{\dOi}{14(1:5)2018}
\lmcsheading{}{1--\pageref{LastPage}}{}{}%
{Feb.~02, 2017}{Jan.~16, 2018}{}

\usepackage[utf8]{inputenc}
\usepackage[T1]{fontenc}
\usepackage{amsmath,amssymb,algorithm}
\usepackage[noend]{algorithmic}
\usepackage{hyperref}
\hypersetup{hidelinks}

\floatname{algorithm}{Alg.}

\providecommand*{\seq}[3]{\ensuremath{#1_{#2}, \dotsc, #1_{#3}}}
\providecommand*{\nat}[0]{\ensuremath{\mathbb{N}}}
\providecommand*{\abs}[1]{\ensuremath{\lvert #1 \rvert}}
\providecommand*{\SBox}[0]{\ensuremath{{\scriptstyle \Box}}}

\DeclareMathOperator{\sol}{sol}
\DeclareMathOperator{\rk}{rk}
\DeclareMathOperator{\push}{push}
\DeclareMathOperator{\syn}{alph}
\DeclareMathOperator{\supp}{supp}
\DeclareMathOperator{\unw}{unw}
\DeclareMathOperator{\ren}{ren}

\allowdisplaybreaks

\ACMCCS{[{\bf Theory of computation}]: Formal languages and automata
  theory --- Tree languages; [{\bf Theory of computation}]: Formal
  languages and automata --- Automata extensions --- Quantitative
  automata}

\def\eg{{\emph{e.g.}}}
\def\ie{{\emph{i.e.}}}

\begin{document}

\title[Pushing for weighted tree automata]
{Pushing for weighted tree automata\rsuper* \\ --- dedicated to the
    memory of \textsc{Zolt\'an \'Esik} (1951--2016) ---}
\titlecomment{{\lsuper*}This is a revised and extended version of
  [\textsc{Maletti, Quernheim}: \emph{Pushing for weighted 
    tree automata}. Proc.\@ 36th Int.\@ Conf.\@ Mathematical
  Foundations of Computer Science, LNCS~6907, p.~460--471, 2011].} 

\author[Th.~Hanneforth]{Thomas Hanneforth}
\address{Universit\"at Potsdam, Human Sciences Faculty, Department
  Linguistik \hspace{0.25\linewidth} \linebreak
  Karl-Liebknecht-Str.~24--25, 14476 Potsdam, Germany} 
\email{thomas.hanneforth@uni-potsdam.de}

\author[A.~Maletti]{Andreas Maletti}
\address{Universit\"at Leipzig, Faculty of Mathematics and Computer
  Science, Institute of Computer Science
  PO~box~100\,920, 04009 Leipzig, Germany}
\email{maletti@informatik.uni-leipzig.de}

\author[D.~Quernheim]{Daniel Quernheim}
\address{Universit\"at Stuttgart, Institute for Natural Language
  Processing \hspace{.3\linewidth} \linebreak
  Pfaffenwaldring~5b, 70569 Stuttgart, Germany}
\email{daniel.quernheim@ims.uni-stuttgart.de}
\thanks{Financially supported by the German Research Foundation~(DFG)
  grant MA\,/\,4959\,/\,1-1.}

\begin{abstract}
  A weight normalization procedure, commonly called pushing, is
  introduced for weighted tree automata~(wta) over commutative
  semifields.  The normalization preserves the recognized weighted
  tree language even for nondeterministic wta, but it is most useful
  for bottom-up deterministic wta, where it can be used for
  minimization and equivalence testing.  In both applications a
  careful selection of the weights to be redistributed followed by
  normalization allows a reduction of the general problem to the
  corresponding problem for bottom-up deterministic unweighted tree
  automata.  This approach was already successfully used by
  \textsc{Mohri} and \textsc{Eisner} for the minimization of
  deterministic weighted string automata.  Moreover, the new
  equivalence test for two wta $M$~and~$M'$ runs in time~$\mathcal
  O \bigl((\abs M + \abs{M'}) \log {(\abs Q + \abs{Q'})} \bigr)$,
  where $Q$~and~$Q'$ are the states of $M$~and~$M'$, respectively,
  which improves the previously best run-time~$\mathcal O \bigl(\abs M
  \cdot \abs{M'} \bigr)$.
\end{abstract}

\maketitle

\section{Introduction}
Weighted tree automata~\cite{fulvog09} have recently found various
applications in fields as diverse as natural language and XML
processing~\cite{knimay09}, system verification~\cite{jac11}, and
pattern recognition.  Most applications require efficient algorithms
for basic manipulations of tree automata such as
determinization~\cite{bucmayvog10}, inference~\cite{mayknivog10}, and 
minimization~\cite{hogmalmay08,hogmalmay07d}.  For example, in the
system verification domain the properties to be verified are typically
easily expressed as a formula in a logic.  It is well-known~\cite{thawri68}
that tree automata are as expressive as monadic second-order logic
with two successors.  This celebrated result was recently generalized
to the weighted setting for various weight
structures~\cite{drovog06,man08,drogotmarmei11,vogdroher16}, so
quantitative specifications are readily available.  However, one of
the main insights gained in the development of the \textsc{mona}
toolkit~\cite{klamol01} (or the \textsc{spass}
system~\cite{weidimfiekumsudwis09}) was that the transformation of a
formula into an equivalent tree automaton heavily relies on the
minimization of the constructed deterministic tree automata as the
automata otherwise grow far too quickly.  Similarly, a major inference
setup, also used in the synthesis subfield in system verification, is
\textsc{Angluin}'s minimally adequate teacher setup~\cite{ang87}.  In
this setup, the learner is given access to an oracle that correctly
supplies coefficients of trees in the weighted tree language to be
learned, which are called coefficient queries, and certificates that
the proposed weighted tree automaton indeed represents the weighted
tree language to be learned, which are called equivalence queries.  In
implementations of the oracle the latter queries are typically
answered by equivalence tests.

As already mentioned, quantitative models have recently enjoyed a lot
of attention.  For example, in natural language processing, weighted
devices are often used to model probabilities, cost functions, or
other features.  In this contribution, we consider
pushing~\cite{moh97,eis03} for weighted 
tree automata~\cite{berreu82,fulvog09} over commutative
semifields~\cite{hebwei98,gol99}.  Roughly speaking, pushing moves
transition weights along a path.  If the weights are properly
selected, then pushing can be used to canonicalize a (bottom-up)
deterministic weighted tree automaton~\cite{bor04b}.  The obtained
canonical representation has the benefit that it can be minimized
using unweighted minimization, in which the weight is treated as a
transition label.  This strategy has successfully been employed
in~\cite{moh97,eis03} for deterministic weighted (finite-state) string
automata, and similar approaches have been used to minimize sequential
transducers~\cite{cho03} and bottom-up tree
transducers~\cite{friseiman11}.  Here we adapt 
the strategy for tree automata.  In particular, we
improve the currently best minimization algorithm~\cite{mal08e} for
a deterministic weighted tree automaton~$M$ with states~$Q$ from
$\mathcal O \bigl(\abs M \cdot \abs Q \bigr)$ to~$\mathcal O
\bigl(\abs M \log {\abs Q} \bigr)$, which coincides with the
complexity of minimization in the unweighted case~\cite{hogmalmay08}.
The improvement is achieved by a careful selection of the signs of
life~\cite{mal08e}.  Intuitively, a sign of life for a state $q$ is a
context that takes $q$ into a final state.  In~\cite{mal08e} the signs
of life are computed by a straightforward exploration algorithm, which
is very efficient, but does not guarantee that states that are later
checked for equivalence receive the same sign of life.  During the
(pair-wise) equivalence checks in~\cite{mal08e} the evaluation of the
weight of a state in the sign of life of another state thus becomes
unavoidable, which causes the increased complexity.  In this
contribution, we precompute an equivalence relation, which, in
general, is still coarser than the state equivalence to be determined,
but equivalent states in this equivalence permit the same sign of
life.  Then we determine a sign of life for each equivalence class.
Later we only refine this equivalence relation to obtain the state
equivalence, so each state will only be evaluated in its sign of life
and this evaluation can be precomputed.  Moreover, the weights
obtained in this evaluation, also called pushing weights, allow a
proper canonicalization in the sense that equivalent states will have
exactly the same weights on corresponding transitions after pushing.
This property sets our algorithm apart from Algorithm~1
of~\cite{mal08e} and allows us to rely on unweighted
minimization~\cite{hogmalmay08}.  Our pushing procedure, which is
defined for general (potentially nondeterministic) weighted tree
automata, always preserves the semantics, so it might also be useful
in other setups.

Secondly, we apply pushing to the problem of testing equivalence.  The
currently fastest algorithm~\cite{drehogmal09b} for checking
equivalence of two deterministic weighted tree automata
$M$~and~$M'$ runs in time~$\mathcal O \bigl(\abs M \cdot \abs{M'}
\bigr)$.  It is well known that two minimal deterministic weighted
tree automata $M$~and~$M'$ are equivalent if and only if they can be
obtained from each other by a pushing operation (with proper pushing
weights).  In other words, equivalent automata $M$~and~$M'$ have the
same transition structure, but their transition weights can differ by
consistent factors.  We extend our approach to minimization also to
equivalence testing, so we again carefully determine the pushing
weight and the sign of life of each state~$q$ of~$M$ such that it
shares the sign of life with all equivalent states of~$M$ but also
with all corresponding states in~$M'$.  This allows us to minimize
both input automata and then treat the obtained automata as unweighted
automata and test them for isomorphism.  This approach reduces the
run-time complexity to~$\mathcal O \bigl((\abs M + \abs{M'}) \log
{(\abs Q + \abs{Q'})} \bigr)$, where $Q$~and~$Q'$ are the states of
$M$~and~$M'$, respectively. 

\section{Preliminaries}
\label{sec:Prel}
We write~$\nat$ for the set of all nonnegative integers and $[1, u]$~for
its subset $\{ i \mid 1 \leq i \leq u\}$ given $u \in \nat$.  The
$k$-fold \textsc{Cartesian} product of a set~$Q$ is 
written as~$Q^k$, and the empty tuple~$() \in Q^0$ is often written
as~$\varepsilon$.  Every finite and nonempty set is also called
alphabet, of which the elements are called symbols.  A ranked
alphabet~$(\Sigma, \mathord{\rk})$ consists of an alphabet~$\Sigma$
and a mapping~$\mathord{\rk} \colon \Sigma \to \nat$, which assigns a
rank to each symbol.  If the ranking~`$\rk$' is obvious from the
context, then we simply write~$\Sigma$ for the ranked alphabet.  For
each $k \in \nat$, we let~$\Sigma_k$ be the set $\{\sigma \in \Sigma
\mid \rk(\sigma) = k\}$ of $k$-ary symbols of~$\Sigma$.  Moreover, we
let $\Sigma(Q) = \{ \sigma w \mid \sigma \in \Sigma,\, w \in
Q^{\rk(\sigma)}\}$.  The set~$T_\Sigma(Q)$ of all $\Sigma$-trees
indexed by~$Q$ is inductively defined to be the smallest set~$T$ such
that $Q \subseteq T$ and $\Sigma(T) \subseteq T$.  Instead
of~$T_\Sigma(\emptyset)$ we simply write~$T_\Sigma$.  The size~$\abs
t$ of a tree~$t \in T_\Sigma(Q)$ is inductively defined by $\abs q =
1$ for every $q \in Q$ and $\abs{\sigma(\seq t1k)} = 1 + \sum_{i =
  1}^k \abs{t_i}$ for every $k \in \nat$, $\sigma \in \Sigma_k$, and
$\seq t1k \in T_\Sigma(Q)$.  To increase readability, we often omit
quantifications like ``for all $k \in \nat$'' if they are obvious from
the context. 

We reserve the use of a special symbol~$\SBox$ that is not an element
in any considered alphabet.  Its function is to mark a designated
position in certain trees called contexts.  Formally, the
set~$C_\Sigma(Q)$ of all $\Sigma$-contexts indexed by~$Q$ is defined
as the smallest set~$C$ such that $\SBox \in C$ and $\sigma(\seq
t1{i-1}, c, \seq t{i+1}k) \in C$ for every $\sigma \in \Sigma_k$,
$\seq t1k \in T_\Sigma(Q)$, $i \in [1, k]$, and $c \in C$.  As before,
we simplify~$C_\Sigma(\emptyset)$ to~$C_\Sigma$.  In simple words, a
context is a tree, in which the special symbol~$\SBox$ occurs exactly
once and at a leaf position.  Note that $C_\Sigma(Q) \cap T_\Sigma(Q)
= \emptyset$, but $C_\Sigma(Q) \subseteq T_\Sigma(Q \cup \{\SBox\})$,
which allows us to treat contexts like trees.  Given $c \in
C_\Sigma(Q)$ and $t \in T_\Sigma(Q \cup \{\SBox\})$, the tree~$c[t]$
is obtained from~$c$ by replacing the unique occurrence of~$\SBox$
in~$c$ by~$t$.  In particular, $c[c'] \in C_\Sigma(Q)$ given that $c,
c' \in C_\Sigma(Q)$.

A commutative semiring~\cite{hebwei98,gol99} is a tuple $(S,
\mathord+, \mathord\cdot, 0, 1)$ such that $(S, \mathord+, 0)$ and
$(S, \mathord\cdot, 1)$ are commutative monoids and $s \cdot 0 = 0$
and $s \cdot (s_1 + s_2) = (s \cdot s_1) + (s \cdot s_2)$ for all $s,
s_1, s_2 \in S$ (\ie, $\cdot$~distributes over~$+$).  It is a
commutative semifield if $(S \setminus \{0\}, \mathord\cdot, 1)$ is a
commutative group (\ie, in addition, for every $s \in S \setminus
\{0\}$ there exists $s^{-1} \in S$ such that $s \cdot s^{-1} = 1$).
% Finally, it is (multiplicatively) cancellative if $s \cdot s_1 = s
% \cdot s_2$ implies $s_1 = s_2$ for all $s, s_1, s_2 \in S$ such that
% $s \neq 0$.  Note that each commutative semifield is cancellative.
Typical commutative semifields include
\begin{itemize}
\item the \textsc{Boolean} semifield $\mathbb B = (\{0, 1\},
  \mathord{\max}, \mathord{\min}, 0, 1)$,
% \item the semiring~$(\nat, \mathord{+}, \mathord{\cdot}, 0, 1)$ of
%   natural numbers,
\item the field~$(\mathbb{Q}, \mathord{+}, \mathord{\cdot}, 0, 1)$ of
  rational numbers, and 
\item the \textsc{Viterbi} semifield $(\mathbb{Q}_{\geq 0},
  \mathord{\max}, \mathord{\cdot}, 0, 1)$, where $\mathbb{Q}_{\geq 0}
  = \{ q \in \mathbb{Q} \mid q \geq 0\}$.
\end{itemize}
Given a mapping $f \colon A \to S$, we write $\supp(f)$ for the set
$\{a \in A \mid f(a) \neq 0\}$ of elements that are mapped via~$f$ to
a non-zero semiring element.  

\begin{quote}
  \emph{For the rest of the paper, let $(S, \mathord+, \mathord\cdot,
    0, 1)$ be a commutative semifield.}\footnote{Clearly, weighted
    tree automata can also be defined for semirings or even more
    general weight structures, but already minimization for
    deterministic finite-state string automata becomes NP-hard for
    simple semirings that are not semifields
    (see~\cite[Section~3]{eis03}).}
\end{quote}

A weighted tree
automaton~\cite{bozlou83,boz99,kui98,borvog03,bor04b,fulvog09} (for
short: wta) is a tuple $M = (Q, \Sigma, \mu, F)$, in which
\begin{itemize}
\item $Q$~is an alphabet of states,
\item $\Sigma$~is a ranked alphabet of input symbols,
\item $\mu \colon \Sigma(Q) \times Q \to S$ assigns a weight to each
  transition, and 
\item $F \subseteq Q$ is a set of final states.
\end{itemize}
We often write elements of~$T_\Sigma(Q) \times Q$ as~$t \to q$ instead
of~$(t, q)$.  The size~$\abs M$ of the wta~$M$ is
\[ \abs M = \sum_{t \to q \in \supp(\mu)} (\abs t + 1) \enspace. \] We
extend the transition weight assignment~$\mu$ to a mapping $h_\mu
\colon T_\Sigma(Q) \times Q \to S$ by
\begin{align*}
  h_\mu(p \to q) &=
  \begin{cases}
    1 & \text{if } p = q \\
    0 & \text{otherwise}
  \end{cases} \\
  h_\mu(\sigma(\seq t1k) \to q) &= \sum_{\seq q1k \in Q} \mu(\sigma(\seq
  q1k) \to q) \cdot \prod_{i = 1}^k h_\mu(t_i \to q_i)
\end{align*}
for all $p, q \in Q$, $\sigma \in \Sigma_k$, and $\seq t1k \in
T_\Sigma(Q)$.  The wta~$M$ recognizes the weighted tree language~$M
\colon T_\Sigma \to S$ such that $M(t) = \sum_{q \in F} h_\mu(t \to
q)$ for every $t \in T_\Sigma$.  Two wta $M$~and~$M'$ are equivalent
if their recognized weighted tree languages coincide.  The unweighted
(finite-state) tree
automaton~\cite{gecste84,gecste97,comdaugillodjaclugtistom07} (for
short: fta) corresponding to~$M$ is $\unw(M) = (Q, \Sigma, \supp(\mu),
F)$.\footnote{An fta computes in the same manner as a wta over
  the \textsc{Boolean} semifield~$\mathbb B$.}  We note that $\supp(M)
\subseteq L(\unw(M))$, where $L(\unw(M))$~is the tree language
recognized by the fta~$\unw(M)$. 

The wta~$M = (Q, \Sigma, \mu, F)$ is \emph{(bottom-up) deterministic}
(or a dwta) if for every $t \in \Sigma(Q)$ there exists at most one~$q
\in Q$ such that $t \to q \in \supp(\mu)$.  In other words, a wta~$M$
is deterministic if and only if $\unw(M)$ is bottom-up deterministic.
In a dwta we can (without loss of information) treat $\mu$~and~$h_\mu$
as partial mappings $\mu \colon \Sigma(Q) \dasharrow Q \times S$ and
$h_\mu \colon T_\Sigma(Q) \dasharrow Q \times S$.  We use
$\mu^{(1)}$~and~$\mu^{(2)}$ as well as $h_\mu^{(1)}$~and~$h_\mu^{(2)}$
for the corresponding projections to the first and second output
component, respectively (\eg, $\mu^{(1)} \colon \Sigma(Q) \dasharrow
Q$ and $\mu^{(2)} \colon \Sigma(Q) \dasharrow S$).  To avoid
complicated distinctions, we treat undefinedness like a value (\ie,
it is equal to itself, but different from every other value).  We
observe that $\supp(M) = L(\unw(M))$ for each dwta~$M$.\footnote{The
  statement holds because each commutative semifield is zero-divisor
  free~\protect{\cite[Lemma~1]{bor03}}.}  Moreover, the restriction to
final states instead of final weights in the definition of a wta does
not restrict the expressive power~\cite[Lemma~6.1.4]{bor04b}, which
applies to both wta and dwta.  In addition, the transformation of a
wta with final weights into an equivalent wta with final states does
not add additional states, so all our results also apply to wta with
final weights.

An equivalence relation~$\equiv$ on a set~$A$ is a reflexive,
symmetric, and transitive subset of~$A^2$.  The equivalence
class (or block)~$[a]_{\mathord{\equiv}}$ of the element~$a \in A$ is
$\{ a' \in A \mid a \equiv a'\}$, and we let $(A' / \mathord{\equiv})
= \{ [a']_{\mathord{\equiv}} \mid a' \in A'\}$ for every $A' \subseteq
A$.  Whenever~$\equiv$ is obvious from the context, we simply omit it.
The equivalence~$\equiv$ respects a set~$A' \subseteq A$ if $[a]
\subseteq A'$ or $[a] \subseteq A \setminus A'$ for every $a \in A$
(\ie, each equivalence class is either completely in~$A'$ or
completely outside~$A'$).

Let $M = (Q, \Sigma, \mu, F)$ be a dwta.  An equivalence relation
$\mathord{\equiv} \subseteq Q^2$ is a congruence (of~$M$) if
$\mu^{(1)}(\sigma(\seq q1k)) \equiv \mu^{(1)}(\sigma(\seq{q'}1k))$ for
every $\sigma \in \Sigma_k$ and all equivalent states $q_1 \equiv
q'_1, \dotsc, q_k \equiv q'_k$.  Note that this definition of
congruence completely disregards the weights, which yields that
$\equiv$~is a congruence for~$M$ if and only if $\equiv$~is a
congruence for~$\unw(M)$.  Two states $q_1, q_2 \in Q$ are weakly
equivalent, written as $q_1 \sim_M q_2$, if $h_\mu^{(1)}(c[q_1]) \in F$
if and only if $h_\mu^{(1)}(c[q_2]) \in F$ for all contexts $c \in
C_\Sigma(Q)$.  In other words, weak equivalence coincides with
classical equivalence~\cite[Definition~II.6.8]{gecste84}
for~$\unw(M)$.  Consequently, the weak equivalence relation~$\sim_M$
is actually a congruence of~$M$ that
respects~$F$~\cite[Theorem~II.6.10]{gecste84}.  The weak equivalence
relation~$\sim_M$ can be computed in time~$\mathcal O \bigl(\abs M \log
{\abs Q} \bigr)$~\cite{hogmalmay08}.  Finally, two states
are (strongly) equivalent, written as $q_1 \equiv_M q_2$ if there
exists a factor $s \in S \setminus \{0\}$ such that for all $c \in
C_\Sigma(Q)$ we have 
\[ h_\mu^{(2)}(c[q_1]) \cdot \chi_F \bigl(h_\mu^{(1)}(c[q_1]) \bigr) =
s \cdot h_\mu^{(2)}(c[q_2]) \cdot \chi_F \bigl(h_\mu^{(1)}(c[q_2])
\bigr) \enspace, \] 
where $\chi_F \colon Q \to \{0, 1\}$ is the characteristic function
of~$F$; \ie, $F(q) = 1$ if and only if $q  \in F$ for all $q \in Q$.
The equivalence relation~$\equiv_M$ is called the
\textsc{Myhill-Nerode} equivalence relation~\cite[Definition~3]{mal08e}.
It is also a congruence that respects~$F$~\cite[Lemma~4]{mal08e}.  If
$M$~is clear from the context, then we just write~$\equiv$ instead
of~$\equiv_M$.

\section{Signs of life}
\label{sec:sol}
First, we demonstrate how to efficiently compute signs of life
(Definition~\ref{df:SoL}), which are evidence that a final state can
be reached.  Together with these signs of life we also compute a
pushing weight for each state (Section~\ref{sec:Push}).  Our
Algorithm~\ref{alg:sol} is a straightforward extension
of~\cite[Algorithm~1]{mal08e} that computes on equivalence classes of
states (with respect to a congruence that respects finality) instead
of states.\footnote{Note that our algorithm is not simply the
  previous algorithm executed on the quotient dwta with respect to the
  congruence.  The original dwta is used essentially in the
  computation of the pushing weights.} This change guarantees that
equivalent states receive the same sign of life, which is an essential
requirement for the algorithms in Sections
\ref{sec:Min}~and~\ref{sec:equiv}.  

Before we start we need to recall the definition of a sign of
life~\cite{mal08e}.  In addition, we recall the relevant properties
that we use in our algorithm.  For the rest of this section, let $M =
(Q, \Sigma, \mu, F)$ be a dwta.

\begin{defi}[{\protect{\cite[Section~2]{mal08e}}}]
  \label{df:SoL}
  A context~$c \in C_\Sigma(Q)$ is a \emph{sign of life} for the
  state~$q \in Q$ if $h_\mu^{(1)}(c[q]) \in F$.  Any state that
  has a sign of life is \emph{live}; otherwise it is \emph{dead}.
\end{defi}

The following lemma justifies that we can compute signs of life for
equivalence classes of congruences that respect~$F$ instead of
individual states since all states of such an equivalence class share
the same signs of life.

\begin{lem}[{\protect{see~\cite[Lemma~9]{mal08e}}}]
  \label{lm:help}
  We have $\mathord{\cong} \subseteq \mathord{\sim_M}$ for every
  congruence~$\cong$ that respects~$F$.  In particular,
  $\mathord{\equiv_M} \subseteq \mathord{\sim_M}$.  Moreover, every
  sign of life for~$q \in Q$ is also a sign of life for every~$q' \in
  [q]_{\cong}$.
\end{lem}

\begin{proof}
  It is known that~$\sim_M$ is the coarsest congruence that
  respects~$F$~\cite[Theorem II.6.10]{gecste84}.\footnote{Mind
    that~$\sim_M$ coincides with classical equivalence on~$\unw(M)$
    and that our notion of congruence completely disregards the
    weights.}  Consequently, $\mathord{\cong} \subseteq
  \mathord{\sim_M}$ and $\mathord{\equiv_M} \subseteq
  \mathord{\sim_M}$ since we already remarked that $\equiv_M$~is also
  a congruence that respects~$F$.  Based on the definition of~$\sim_M$
  it is trivial to see that all elements of an equivalence class
  of~$\sim_M$ share the same signs of life~\cite[Lemma~9]{mal08e}.
  Since $[q]_{\cong} \subseteq [q]_{\sim_M}$ we obtain the desired
  statement.
\end{proof}

\begin{algorithm}[t]
  \begin{algorithmic}[2]
    \REQUIRE dwta $M = (Q, \Sigma, \mu, F)$ and
    congruence~$\mathord{\cong} \subseteq Q^2$ of~$M$ that respects~$F$
    \ENSURE return live state partition $(L/\mathord{\cong})$ and the
    mappings $\mathord{\sol} \colon (L/\mathord{\cong}) \to
    C_\Sigma(Q)$ and $\lambda \colon L \to S \setminus \{0\}$ such
    that $\lambda(q) = h_\mu^{(2)}(\sol([q]_{\cong})[q])$ for every $q \in L$
    \smallskip\hrule\smallskip

    \STATE $L \gets (F/\mathord{\cong})$ 
      \COMMENT{final states are trivially live \dots}%
    \smallskip
    \FORALL{$B \in L$}
      \STATE $\sol(B) \gets \SBox$ \label{ln:final1}
        \COMMENT{\dots\ with the trivial context as sign of life\dots}%
      \STATE $\lambda(q) \gets 1$ for all $q \in B$
        \label{ln:final2}
        \COMMENT{\dots\ and trivial pushing weight}%
    \ENDFOR

    \smallskip 
    \STATE $U \gets L$
      \COMMENT{start from the final states}
    \smallskip
 
    \WHILE{$U \neq \emptyset$}
      \STATE take~$B \in U$ and $U \gets U \setminus \{B\}$ 
        \COMMENT{get an unexplored class}% 
      \smallskip
      \FORALL{$\sigma(\seq q1k) \in \Sigma(Q)$ such that
        $\mu^{(1)}(\sigma(\seq q1k)) \in B$} 
        \smallskip
        \FORALL{$i \in [1, k]$ such that $[q_i]_{\cong} \notin L$}
          \STATE $c \gets \sigma(\seq q1{i-1}, \SBox, \seq q{i+1}k)$
            \COMMENT{prepare context}%
          \STATE $L \gets L \cup \{[q_i]_{\cong}\}$;  $U \gets U \cup
            \{[q_i]_{\cong}\}$
            \COMMENT{add class to $L$~and~$U$}%
          \STATE $\sol([q_i]_{\cong}) \gets
            \sol(B)[c]$
            \label{ln:ind1}
            \COMMENT{add transition to target block's sign of life}%
          \STATE $\lambda(q) \gets \lambda(\mu^{(1)}(c[q])) \cdot
            \mu^{(2)}(c[q])$ for all $q \in [q_i]_{\cong}$
            \COMMENT{multiply transition weight}%
            \label{ln:ind2}
        \ENDFOR
      \ENDFOR
    \ENDWHILE
    \smallskip 
    \RETURN{($L, \mathord{\sol}, \lambda)$}
  \end{algorithmic}
  \caption{\textsc{ComputeSoL}: Compute a sign of life and its weight
    for each state.}
  \label{alg:sol}
\end{algorithm}

Algorithm~\ref{alg:sol} simply attempts to reach all states from the
final states computing a context that takes the state to a final state
(\ie, a sign of life) as well as its weight in the process.  Due to
Lemma~\ref{lm:help} the signs of life are computed for equivalence
classes (or blocks) instead of individual states.  Now let us explain
Algorithm~\ref{alg:sol} in detail.  Every final 
state~$q \in F$ is trivially live as evidenced by the trivial sign of
life~$\SBox$.  Since the congruence~$\cong$ respects~$F$, the
set~$(F/\mathord{\cong})$ contains equivalence classes that contain
only final states.  We set the sign of life for each class to~$\SBox$
[see Line~\ref{ln:final1}], and for each involved state~$q$ we set its
pushing weight to~$1$ [see Line~\ref{ln:final2}].  Overall, this
initialization takes time~$\mathcal O \bigl(\abs{F} \bigr)$.  Next, we
add all those blocks to the live states~$L$ and to the blocks~$U$ yet
to be explored.  As long as there are still unexplored blocks, we
select a block~$B$ from~$U$ and remove it from~$U$.  Then we consider
all transitions that end in a state that belongs to the block~$B$ and 
check whether it contains a source state that is not yet present
in~$L$.  For each such source state~$q_i$, we add its equivalence
class~$[q_i]_{\cong}$ to both $L$~and~$U$.  Then we set the sign of
life for this class to the sign of life for~$B$ extended by the
considered transition [see Line~\ref{ln:ind1}].  Finally, we select
each state~$q$ from~$[q_i]_{\cong}$ and compute a pushing weight by
multiplying the weight of the currently considered transition
with~$q_i$ replaced by~$q$ to the already computed pushing weight for
the target state reached by the modified transition [see
Line~\ref{ln:ind2}].

\begin{thm}
  \label{thm:sol}
  Algorithm~\ref{alg:sol} is correct and runs in time~$\mathcal O
  \bigl(\abs M + \abs Q \bigr)$.
\end{thm}

\begin{proof}
  Since our algorithm is similar to the one of~\cite{mal08e}, our
  proof closely resembles the proofs of~\cite[Lemma~10 and
  Theorem~11]{mal08e} adjusted to equivalence classes.  We already
  argued that the initialization runs in time $\mathcal O \bigl(\abs F
  \bigr) \subseteq \mathcal O \bigl(\abs Q \bigr)$.  It is easy to see
  that $U \subseteq L$ at all times in the main loop
  [Line~6--\ref{ln:ind2}] of the algorithm.  Consequently, each block
  can be added at most once to~$U$ since it is added at the same time
  to~$L$ and only blocks not in~$L$ can be added to~$U$.  This yields
  that the main loop executes at most $\abs{(Q/\mathord{\cong})} \leq
  \abs Q$~times.  The inner loop [Line~9--\ref{ln:ind2}] can execute
  at most $\abs M$~times since each transition is considered at most
  once in the middle loop and at most once for each source state of
  the transition.  The statements in the inner loop all execute in
  constant time except for Line~\ref{ln:ind2}, which can be executed
  once for each state~$q \in Q$.  Overall, we thus obtain the running
  time~$\mathcal O \bigl(\abs M + \abs Q \bigr)$.

  Now let us prove the post-conditions.  By Lemma~\ref{lm:help} we
  know that signs of life are shared between elements in an
  equivalence class of~$\cong$.  The remaining statements are proved
  by induction along the outer main loop.  Initially, we set
  \[ \lambda(q) = 1 = h_\mu^{(2)}(q) = h_\mu^{(2)}(\SBox[q]) \] by
  Lines~\ref{ln:final1}--\ref{ln:final2}, which proves the
  post-condition because $\sol([q]_{\cong}) = \SBox$.  In the main
  loop, we set $\lambda(q) = \lambda(\mu^{(1)}(c[q])) \cdot
  \mu^{(2)}(c[q])$ in Line~\ref{ln:ind2}.  The equivalence class of
  $q' = \mu^{(1)}(c[q])$ has already been explored in a previous
  iteration because $q \cong q_i$, which by the congruence property
  yields $\mu^{(1)}(c[q]) \cong \mu^{(1)}(c[q_i])$ and the latter was
  in the explored equivalence class~$B$, which in turn yields that the
  former is in~$B$.  Consequently, we can employ the induction
  hypothesis and obtain $\lambda(q') = h_\mu^{(2)}(\sol(B)[q'])$.  In
  addition,
  \begin{align*}
    \lambda(q) &= \lambda(q') \cdot \mu^{(2)}(c[q]) =
    h_\mu^{(2)}(\sol(B)[q']) \cdot \mu^{(2)}(c[q]) \\
    &= h_\mu^{(2)}(\sol(B)[c[q]]) = h_\mu^{(2)}((\sol(B)[c])[q])
    \enspace,
  \end{align*}
  which proves the post-condition because $\sol([q]_{\cong}) =
  \sol([q_i]_{\cong}) = \sol(B)[c]$ by Line~\ref{ln:ind1}.  Clearly,
  $\sol([q]_{\cong})$ is a sign of life for~$q$, which proves that
  $q$~is live.  Finally, suppose that there is a live state~$q \in Q$
  such that $[q]_{\cong} \notin L$ (\ie, we assume a live state that
  is not classified as such by Algorithm~\ref{alg:sol}).  Since it is
  live, it has a sign of life~$c \in C_\Sigma(Q)$.  By induction
  on~$c$ we can prove that, when processing~$c[q]$, there exists a
  transition that uses a source state~$q_i$ such that $[q_i]_{\cong}
  \notin L$, whereas the target state~$q'$ is such that $[q']_{\cong}
  \in L$.\footnote{Such a switch must exist because all the final
    states are represented in~$L$.} However, since $[q']$~was
  explored, the considered transition was considered in the algorithm,
  which means that the equivalence class~$[q_i]_{\cong}$ was added
  to~$L$.  This contradicts the assumption, which shows that all
  states that are not represented in~$L$ are indeed dead.
\end{proof}

\begin{figure}[t]
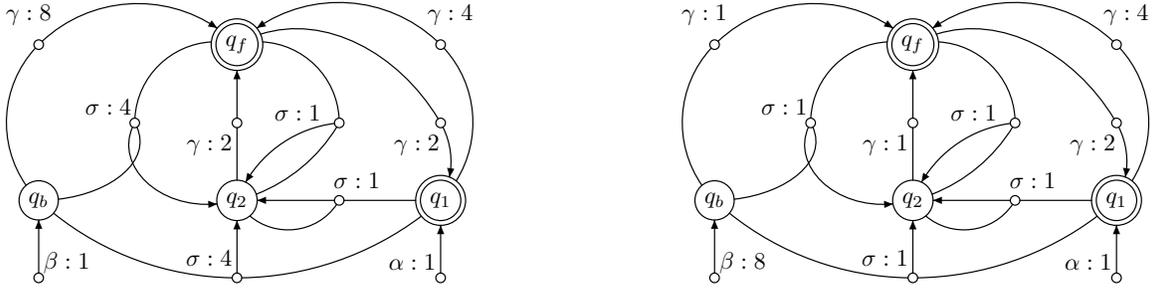

  \begin{center}
    \includegraphics[scale=0.9]{figure/example1.mps}
    \hfill
    \includegraphics[scale=0.9]{figure/example2.mps}
    \caption{Dwta over the rational numbers before (left) and after
      (right) pushing.}
    \label{fig.before}
    \label{fig.after}
  \end{center}
\end{figure}

\begin{exa}
  \label{ex:Example}
  Our example dwta~$N = (Q, \Sigma, \mu, F)$ is depicted left in
  Figure~\ref{fig.before}.  For any transition (small circle, the
  annotation specifies the input symbol and the weight separated by a
  colon), the arrow leads to the target state and the source
  states~$\seq q1k$ have been arranged in a counter-clockwise fashion
  starting from the target arrow.  For example, the bottom center
  transition labeled~$\sigma : 4$ in the left dwta of
  Figure~\ref{fig.before} corresponds to~$\mu(\sigma(q_b, q_1) \to
  q_2) = 4$; \ie, its target state is~$q_2$, its symbol is~$\sigma$,
  its source states are~$q_bq_1$ (in this order), and its weight
  is~$4$.  As usual, final states are doubly circled.  The graphical
  representation of wta is explained in detail in~\cite{bor04b}.  The
  coarsest congruence~$\cong$ respecting~$F 
  = \{q_1, q_f\}$ is represented by the set $\{\{q_1,q_f\},\,
  \{q_2,q_b\}\}$ of equivalence classes (\ie, partition).  We use 
  this congruence in Algorithm~\ref{alg:sol}.  First, the block~$F$ of
  final states is marked as live and added to~$U$.  It is assigned the
  trivial context~$\SBox$ as sign of life and each final state is
  assigned the trivial weight~$1$.  Clearly, we can only select one
  equivalence class~$B = F$ in the main loop.  Let us consider the
  transition~$\gamma(q_b) \to q_f$, whose target state~$q_f$ is
  in~$B$.  Since $[q_b]_{\cong} = \{q_2, q_b\}$~has not yet been
  marked as live, we add it to both $L$~and~$U$.  In addition, we set
  its sign of life to~$\gamma(\SBox)$.  Finally, we set the pushing
  weights to $\lambda(q_b) = \lambda(q_f) \cdot \mu^{(2)}(\gamma(q_b)
  \to q_f) = 8$ and $\lambda(q_2) = \lambda(q_f) \cdot
  \mu^{(2)}(\gamma(q_2) \to q_f) = 2$.  Now all states are live, so
  the loops will terminate.  Consequently, we have computed all signs
  of life and the pushing weights
  \[ \lambda(q_1) = \lambda(q_f) = 1 \qquad \lambda(q_2) = 2 \qquad
  \text{and} \qquad \lambda(q_b) = 8 \enspace. \]
\end{exa}

\section{Pushing}
\label{sec:Push}
The \textsc{Myhill-Nerode} congruence requires that there is a unique
scaling factor for every pair~$(q, q')$ of equivalent states.  Thus,
any fixed sign of life~$c$ for both $q$~and~$q'$ [for which
$\chi_F(h_\mu^{(1)}(c[q])) = 1 = \chi_F(h_\mu^{(1)}(c[q']))$] yields non-zero
weights $h_\mu^{(2)}(c[q])$~and~$h_\mu^{(2)}(c[q'])$, which can be used to
determine this unique scaling factor between $q$~and~$q'$.  In fact,
we already computed those weights $\lambda(q)$~and~$\lambda(q')$ in
Algorithm~\ref{alg:sol}.  By Lemma~\ref{lm:help}, states that are not
weakly equivalent (and thus might not have the same sign of life after 
executing Algorithm~\ref{alg:sol} with~$\sim_M$) also cannot be
equivalent.  For the remaining pairs of live
states, we computed a sign of life~$\sol([q]_{\sim_M})$ for the
equivalence class~$[q]_{\sim_M}$ of~$q$ in the previous section.  In
addition, we computed pushing weights $\lambda(q)$ and $\lambda(q')$.
Now, we will use these weights to normalize the wta by
\emph{pushing}~\cite{moh97,eis03,posgil09}.  Intuitively, pushing
cancels the scaling factor for equivalent states, which we will prove
in the next section.  In general, it just redistributes weights along
the transitions.  In weighted (finite-state) string
automata~\cite{sak09}, pushing is performed from the final states
towards the initial states~\cite{moh97}.  Since we work with bottom-up
wta~\cite{bor04b} (\ie, our notion of determinism is bottom-up), this
works analogously here by moving weights from the root towards the
leaves.  However, we introduce our notion of pushing for arbitrary,
not necessarily deterministic wta.  To this end, we lift the
corresponding definition~\cite[page~296]{moh97} from string to tree
automata.

\begin{quote}
  \emph{In this section, let~$M = (Q, \Sigma, \mu, F)$ be an arbitrary
    wta and $\lambda \colon Q \to S \setminus \{0\}$ be an arbitrary
    mapping such that $\lambda(q) = 1$ for every $q \in
    F$.}
\end{quote}

\begin{defi}
  \label{df:push}
  The \emph{pushed} wta $\push_\lambda(M)$ is $(Q, \Sigma, \mu', F)$
  such that for every $\sigma \in \Sigma_k$ and $q, \seq q1k \in Q$
  \[ \mu'(\sigma(\seq q1k) \to q) = \lambda(q) \cdot \mu(\sigma(\seq
  q1k) \to q) \cdot \prod_{i = 1}^k \lambda(q_i)^{-1} \enspace. \]
\end{defi}

The mapping~$\lambda$ indicates the pushed weights.  It is non-zero
everywhere and has to be~$1$ for final states because our model does
not have final weights.\footnote{As already mentioned, the restriction
to final states is a convenience and not an essential restriction.}
In the pushed wta~$\push_\lambda(M)$, the weight of every transition
leading to the state~$q \in Q$ is obtained from the weight of the
corresponding transition in~$M$ by multiplying the
weight~$\lambda(q)$.  To compensate, the weight of every transition
leaving the state~$q$ will cancel the weight~$\lambda(q)$ by
multiplying with~$\lambda(q)^{-1}$.  Thus, we expect an equivalent wta
after pushing, which we confirm by showing that
$M$~and~$\push_\lambda(M)$ are indeed equivalent.  The corresponding
statement for string automata is~\cite[Lemma~4]{moh97}.

\begin{prop}
  \label{prop:pushEquiv}
  The wta $M$~and~$\push_\lambda(M)$ are equivalent.  Moreover, if $M$
  is deterministic, then so is~$\push_\lambda(M)$.
\end{prop}

\begin{proof}
  Let $\push_\lambda(M) = M' = (Q, \Sigma, \mu', F)$.  The preservation of
  determinism is obvious because $\supp(\mu') \subseteq
  \supp(\mu)$.\footnote{In fact, $\supp(\mu') = \supp(\mu)$ because
    semifields are zero-divisor free~\protect{\cite[Lemma~1]{bor03}}.}
  We prove that $h_{\mu'}(t \to q) = \lambda(q) \cdot h_\mu(t \to q)$
  for every $t \in T_\Sigma$ and $q \in Q$ by induction on~$t$.  Let
  $t = \sigma(\seq t1k)$ for some $\sigma \in \Sigma_k$ and $\seq t1k
  \in T_\Sigma$.  By the induction hypothesis, we have $h_{\mu'}(t_i
  \to q_i) = \lambda(q_i) \cdot h_\mu(t_i \to q_i)$ for every $i \in
  [1, k]$ and $q_i \in Q$.  Consequently,
  \begin{align*}
    h_{\mu'}(t \to q) &= \sum_{\seq q1k \in Q} \mu'(\sigma(\seq q1k)
    \to q) \cdot \prod_{i = 1}^k h_{\mu'}(t_i \to q_i) \\*
    &= \sum_{\seq q1k \in Q} \lambda(q) \cdot \mu(\sigma(\seq q1k) \to
    q) \cdot \prod_{i = 1}^k \lambda(q_i)^{-1} \cdot \prod_{i = 1}^k
    \Bigl( \lambda(q_i) \cdot h_\mu(t_i \to q_i) \Bigr) \\*
    &= \lambda(q) \cdot h_\mu(t \to q) \enspace.
  \end{align*}
  We complete the proof as follows.
  \begin{align*}
    M'(t) &= \sum_{q \in F} h_{\mu'}(t \to q) =
    \sum_{q \in F} \lambda(q) \cdot h_\mu(t \to q) = \sum_{q
      \in F} h_\mu(t \to q) = M(t) \end{align*}
  because $\lambda(q) = 1$ for every $q \in F$.
\end{proof}

\begin{thm}
  \label{thm:push}
  The wta $\push_\lambda(M)$ is equivalent to~$M$ and can be obtained
  in time~$\mathcal O \bigl(\abs M \bigr)$.
\end{thm}

\begin{exa}
  Let us return to our example dwta~$N$ left in
  Figure~\ref{fig.before} and perform pushing.  The pushing
  weights~$\lambda$ are given in Example~\ref{ex:Example}.  We
  consider the transition~$\sigma(q_b, q_f) \to q_2$, which has
  weight~$4$ in~$N$.  In~$\push_\lambda(N)$ this transition has the weight
  \[ \lambda(q_2) \cdot \mu(\sigma(q_b, q_f) \to q_2) \cdot
  \lambda(q_b)^{-1} \cdot \lambda(q_f)^{-1} = 2 \cdot 4 \cdot 8^{-1}
  \cdot 1^{-1} = 1 \enspace. \] The dwta~$\push_\lambda(N)$ is
  presented right in Figure~\ref{fig.after}.  With a little effort, we
  can confirm that $q_2$~and~$q_b$ are equivalent
  in~$\push_\lambda(N)$, whereas $q_1$~and~$q_f$ are not.
\end{exa}

\section{Minimization}
\label{sec:Min}
Our main application of weight pushing is efficient dwta minimization,
which we present next.  The overall structure of our 
minimization procedure is presented in Algorithm~\ref{alg:Overall}.
As mentioned earlier, the coarsest congruence~$\sim_M$ for a dwta $M =
(Q, \Sigma, \mu, F)$ that respects~$F$ can be obtained by
minimization~\cite{hogmalmay08} of~$\unw(M)$.  We call this procedure
$\text{\textsc{ComputeCoarsestCongruence}}$ and supply it with a
dwta~$M$ and an equivalence relation.  It returns the coarsest
congruence (of~$M$) that refines the given equivalence relation.

\begin{quote}
  \emph{Let $M = (Q, \Sigma, \mu, F)$ be a dwta (without useless
    states) and $\lambda \colon Q \to S \setminus \{0\}$ be the
    pushing weights computed by Algorithm~\ref{alg:sol} when run on
    $M$~and~$\sim_M$.\footnote{In a dwta without useless states we
      have~$\abs Q \leq \abs M$.}  In addition, we let
    $\push_\lambda(M) = M' = (Q, \Sigma, \mu', F)$.}
\end{quote}

The dwta~$M'$ has the property that $(\mu')^{(2)}(\sigma(\seq q1k)) =
(\mu')^{(2)}(\sigma(\seq{q'}1k))$ for all $\sigma \in \Sigma_k$ and
states $q_i \equiv_M q'_i$ for every $i \in [1, k]$.  We will prove
this property~\eqref{eq:mn} in Lemma~\ref{lm:cong}.  It is this
property, which, in analogy to the string case~\cite{moh97,eis03},
allows us to compute the equivalence~$\mathord{\equiv_M} =
\mathord{\sim_N}$ on an unweighted fta~$N$, in which we treat the
transition weight as part of the input symbol.  For example, the
algorithm of~\cite{hogmalmay08} can then be used to
compute~$\mathord{\sim_N}$.  Finally, we merge the equivalent states 
using the information about the scaling factors contained in the
pushing weights~$\lambda$ in the same way as in~\cite{mal08e}.  Let
us start with the formal definitions.\footnote{We avoid a change of
  the weight structure from our semifield to the \textsc{Boolean}
  semifield~$\mathbb B$ since the multiplicative submonoid induced
  by~$\{0, 1\}$ is isomorphic to the multiplicative monoid of~$\mathbb
  B$.  Thus, our dwta with weights in~$\{0, 1\}$ compute in the same
  manner as a dwta over~$\mathbb B$ or equivalently a deterministic
  fta.}

\begin{algorithm}[t]
  \begin{algorithmic}[2]
    \REQUIRE a dwta~$M$ with states~$Q$
    \ENSURE return a minimal, equivalent dwta
    \smallskip\hrule\smallskip
    \STATE $\mathord{\sim_M} \gets
    \textsc{ComputeCoarsestCongruence}(M, Q \times Q)$
      \COMMENT{complexity: $\mathcal O \bigl(\abs M \log {\abs Q}
        \bigr)$}%
    \STATE $(L, \mathord{\sol}, \lambda) \gets \textsc{ComputeSoL}(M,
    \mathord{\sim_M})$
      \COMMENT{complexity: $\mathcal O \bigl(\abs M \bigr)$}%
    \STATE $N \gets \syn(\push_\lambda(M))$
      \COMMENT{complexity: $\mathcal O \bigl(\abs M \bigr)$}%
    \STATE $\mathord{\equiv_M} \gets
    \textsc{ComputeCoarsestCongruence}(N, \mathord{\sim_M})$ 
      \COMMENT{complexity: $\mathcal O \bigl(\abs M \log {\abs Q}
        \bigr)$}%
    \RETURN $\textsc{MergeStates}(M, \mathord{\equiv_M}, \lambda)$
      \COMMENT{complexity: $\mathcal O \bigl(\abs M \bigr)$}%
  \end{algorithmic}
  \caption{{\protect{Overall structure of our minimization algorithm;
        see~\protect{\cite{mal08e}} for details on the procedure
        \textsc{MergeStates}.  Note that the final merging is
        performed on the input dwta~$M$.  The alphabetic dwta~$N$ is
        only needed to compute the equivalence~$\equiv_M$.}}}
  \label{alg:Overall}
\end{algorithm}

\begin{defi}
  \label{df:Alp}
  Let $M = (Q, \Sigma, \mu, F)$ be a dwta, and let $S' = \{ \mu(\tau)
  \mid \tau \in \supp(\mu) \}$ be the finite set of non-zero weights
  that occur as transition weights in~$M$.  The \emph{alphabetic
    dwta}~$\syn(M)$ for~$M$ is $(Q, \Sigma \times S', \mu'', F)$,
  where
  \begin{itemize}
  \item $\rk(\langle \sigma, s\rangle) = \rk(\sigma)$ for every
    $\sigma \in \Sigma$ and $s \in S'$,
  \item $\mu''(\tau) = 1$ for every $\tau \in \supp(\mu'')$, and
  \item for every $\sigma \in \Sigma_k$, $s \in S'$, and $q, \seq q1k \in
    Q$
    \[ \mu''(\langle \sigma, s\rangle(\seq q1k) \to q) = 1 \quad \iff
    \quad \mu(\sigma(\seq q1k) \to q) = s \enspace. \] 
  \end{itemize}
\end{defi}

% We will also need the inverse of the operation above.  The next
% definition states the requirements, which we need to perform the
% inverse construction.  

% \begin{definition}
%   \label{df:InvAlp}
%   Let $M'' = (Q, \Sigma \times S', \mu'', F)$ be a dwta such that $S'
%   \subseteq S$.  The dwta is \emph{alphabeticized} if
%   \begin{itemize}
%   \item $\rk(\langle \sigma, s\rangle) = \rk(\sigma)$ for every
%     $\sigma \in \Sigma$ and $s \in S'$,
%   \item $\mu''(\tau) = 1$ for every $\tau \in \supp(\mu'')$, and
%   \item for every $\sigma \in \Sigma_k$ and $\seq q1k \in Q$ there
%     exists at most one $q \in Q$ and $s \in S'$ such that
%     \[ \mu''(\langle \sigma, s\rangle(\seq q1k) \to q) = 1 \enspace. \] 
%   \end{itemize}
% \end{definition}

% Clearly, $\syn(M)$~is alphabeticized for every dwta~$M$.  An
% alphabeticized dwta over the extended ranked alphabet~$\Sigma \times
% S'$ can be turned back into a dwta over~$\Sigma$.  We will
% write~$\syn^{-1}$ for this operation.

Clearly, the construction of~$\syn(M)$ can be performed in
time~$\mathcal O \bigl(\abs M \bigr)$.  Next, we show that the 
equivalence~$\equiv_M$ in~$M$ coincides with the
equivalence~$\sim_{\syn(M')}$ in~$\syn(M')$, where $M' =
\push_\lambda(M)$.  We achieve this proof by showing both inclusions.

\begin{lem}
  \label{lm:cong}
  The congruence~$\equiv_M$ of~$M$ is a congruence of~$\syn(M')$ that
  respects~$F$.
\end{lem}

\begin{proof}
  Let $\syn(M') = (Q, \Sigma \times S', \mu'', F)$.  Since
  $M$~and~$\syn(M')$ have the same final states~$F$,
  $\equiv_M$~trivially respects~$F$ because it is a congruence
  of~$M$ that respects~$F$.  Naturally, $\equiv_M$ is an
  equivalence, so it remains to prove the congruence property
  for~$\syn(M')$.  Let $\sigma \in \Sigma_k$ and $q_i \equiv_M
  q'_i$ for every $i \in [1, k]$.  Then 
  \[ \mu^{(1)}(\sigma(\seq q1k)) \equiv_M
  \mu^{(1)}(\sigma(\seq{q'}1k)) \] because $\equiv_M$ is a congruence
  of~$M$.  For the moment, let us assume that\footnote{Mind that we
    compare the weights in~$M' = \push_\lambda(M)$ here.}
  \[ (\mu')^{(2)}(\sigma(\seq q1k)) = s =
  (\mu')^{(2)}(\sigma(\seq{q'}1k)) \enspace, \] then
  \begin{align*}
    &\phantom{{}\equiv_M{}} (\mu'')^{(1)}(\langle \sigma, s\rangle(\seq q1k))
    &&= (\mu')^{(1)}(\sigma(\seq q1k)) &&= \mu^{(1)}(\sigma(\seq q1k)) \\
    &\equiv_M \mu^{(1)}(\sigma(\seq{q'}1k)) &&=
    (\mu')^{(1)}(\sigma(\seq{q'}1k)) &&= (\mu'')^{(1)}(\langle \sigma,
    s\rangle(\seq{q'}1k)) \enspace.
  \end{align*} 
  For all the remaining combinations of~$\langle \sigma, s'\rangle$ we
  have that both $(\mu'')^{(1)}(\langle \sigma, s'\rangle(\seq q1k))$ and
  $(\mu'')^{(1)}(\langle \sigma, s'\rangle(\seq{q'}1k))$ are undefined
  and thus equal.  We have thus proved the congruence property given
  the assumption.  Consequently, it remains to show that the assumption
  \begin{equation}
    \label{eq:mn}
    (\mu')^{(2)}(\sigma(\seq q1k)) =
    (\mu')^{(2)}(\sigma(\seq{q'}1k)) 
  \end{equation}
  is true.  By Definition~\ref{df:push}, we have
  \begin{align}
    \label{eq:m3}
    (\mu')^{(2)}(\sigma(\seq q1k)) &= \lambda(\mu^{(1)}(\sigma(\seq
    q1k))) \cdot \mu^{(2)}(\sigma(\seq q1k)) \cdot \prod_{i = 1}^k 
    \lambda(q_i)^{-1} \\ 
    \label{eq:m4}
    (\mu')^{(2)}(\sigma(\seq{q'}1k)) &=
    \lambda(\mu^{(1)}(\sigma(\seq{q'}1k))) \cdot 
    \mu^{(2)}(\sigma(\seq{q'}1k)) \cdot \prod_{i = 1}^k
    \lambda(q'_i)^{-1} \enspace.
  \end{align}
  Now we prove that
  \begin{align}
    \notag
    &\phantom{{}={}} \lambda(\mu^{(1)}(c_j[q_j])) \cdot
    \mu^{(2)}(c_j[q_j]) \cdot \prod_{i = 1}^{j-1} \lambda(q'_i)^{-1}
    \cdot \prod_{i = j}^k \lambda(q_i)^{-1} \\
    \label{eq:mm}
    &= \lambda(\mu^{(1)}(c_j[q'_j])) \cdot \mu^{(2)}(c_j[q'_j]) \cdot
    \prod_{i = 1}^j \lambda(q'_i)^{-1} \cdot \prod_{i = j +1}^k
    \lambda(q_i)^{-1} 
  \end{align}
  for every $j \in [1, k]$, where $c_j = \sigma(\seq{q'}1{j-1}, \SBox,
  \seq q{j+1}k)$.  Let $p_j = \mu^{(1)}(c_j[q_j])$ and $p'_j =
  \mu^{(1)}(c_j[q'_j])$.  Since $q_j \equiv_M q'_j$, we also have that
  $p_j \equiv_M p'_j$ because $\equiv_M$~is a congruence of~$M$.  This
  yields that $p_j \sim_M p'_j$ by Lemma~\ref{lm:help}.  Let $c =
  \sol([p_j]_{\sim_M})$ be a sign of life for both $p_j$~and~$p'_j$.
  Moreover, we have a constant scaling factor between the equivalent
  states $q_j$~and~$q'_j$, which yields
  \begin{align}
    \label{eq:m1}
    \frac{\lambda(q_j)} {\lambda(q'_j)} &\stackrel{(\dagger)}=
    \frac{h_\mu^{(2)}(c[c_j[q_j]])} {h_\mu^{(2)}(c[c_j[q'_j]])}
    \stackrel{(\ddagger)}= \frac{h_\mu^{(2)}(c[p_j]) \cdot \mu^{(2)}(c_j[q_j])}
    {h_\mu^{(2)}(c[p'_j]) \cdot \mu^{(2)}(c_j[q'_j])} \\
    \label{eq:m2}
    \frac{\lambda(p_j)} {\lambda(p'_j)}
    &\stackrel{\phantom{(\dagger)}}= \frac{h_\mu^{(2)}(c[p_j])}
    {h_\mu^{(2)}(c[p'_j])} \enspace,
  \end{align} 
  where $(\dagger)$~holds because $c[c_j]$~is a sign of life for both
  $q_j$~and~$q'_j$ and $(\ddagger)$~holds essentially by definition.
  With these equations, let us inspect the main equality.
  \begin{align*}
    &\phantom{{}={}} \frac{\lambda(\mu^{(1)}(c_j[q_j])) \cdot
      \mu^{(2)}(c_j[q_j]) \cdot \prod_{i = 1}^{j-1} \lambda(q'_i)^{-1}
      \cdot \prod_{i = j}^k \lambda(q_i)^{-1}}
    {\lambda(\mu^{(1)}(c_j[q'_j])) \cdot \mu^{(2)}(c_j[q'_j]) \cdot
      \prod_{i = 1}^j \lambda(q'_i)^{-1} \cdot \prod_{i = j +1}^k
      \lambda(q_i)^{-1} } \\
    &= \frac{\lambda(p_j) \cdot \mu^{(2)}(c_j[q_j]) \cdot \lambda(q_j)^{-1}}
    {\lambda(p'_j) \cdot \mu^{(2)}(c_j[q'_j]) \cdot
      \lambda(q'_j)^{-1}} \stackrel{\eqref{eq:m2}}= 
    \frac{h_\mu^{(2)}(c[p_j]) \cdot \mu^{(2)}(c_j[q_j])}
    {h_\mu^{(2)}(c[p'_j]) \cdot \mu^{(2)}(c_j[q'_j])} \cdot
    \frac{\lambda(q'_j)} {\lambda(q_j)} \stackrel{\eqref{eq:m1}}= 1
  \end{align*}
  
  Now we are ready to return to the proof obligation expressed
  in~\eqref{eq:mn}.  We apply~\eqref{eq:mm} in total $k$~times to
  obtain the desired statement.
  \begin{align*}
    (\mu')^{(2)}(\sigma(\seq q1k)) 
    &\stackrel{\eqref{eq:m3}}= \lambda(\mu^{(1)}(\sigma(\seq
    q1k))) \cdot \mu^{(2)}(\sigma(\seq q1k)) \cdot \prod_{i = 1}^k 
    \lambda(q_i)^{-1} \\
    &\stackrel{\phantom{\eqref{eq:m3}}}= \lambda(\mu^{(1)}(c_1[q_1]))
    \cdot \mu^{(2)}(c_1[q_1]) \cdot \prod_{i = 1}^0 \lambda(q'_i)^{-1}
    \cdot \prod_{i = 1}^k \lambda(q_i)^{-1} \\
    &\stackrel{\eqref{eq:mm}}= \lambda(\mu^{(1)}(c_1[q'_1])) \cdot
    \mu^{(2)}(c_1[q'_1]) \cdot \prod_{i = 1}^1 \lambda(q'_i)^{-1}
    \cdot \prod_{i = 2}^k \lambda(q_i)^{-1} \\
    &\stackrel{\phantom{\eqref{eq:m3}}}= \lambda(\mu^{(1)}(c_2[q_2]))
    \cdot \mu^{(2)}(c_2[q_2]) \cdot \prod_{i = 1}^1 \lambda(q'_i)^{-1}
    \cdot \prod_{i = 2}^k \lambda(q_i)^{-1} \\
    &\dots \\
    &\stackrel{\eqref{eq:mm}}= \lambda(\mu^{(1)}(c_k[q'_k])) \cdot
    \mu^{(2)}(c_k[q'_k]) \cdot \prod_{i = 1}^k \lambda(q'_i)^{-1}
    \cdot \prod_{i = k+1}^k \lambda(q_i)^{-1} \\
    &\stackrel{\phantom{\eqref{eq:m4}}}=
    \lambda(\mu^{(1)}(\sigma(\seq{q'}1k))) \cdot
    \mu^{(2)}(\sigma(\seq{q'}1k)) \cdot \prod_{i = 1}^k
    \lambda(q'_i)^{-1} \\
    &\stackrel{\eqref{eq:m4}}= (\mu')^{(2)}(\sigma(\seq{q'}1k))
      \enspace,
  \end{align*}
  which completes the proof.
\end{proof}

\begin{thm}
  \label{thm:Main}
  We have $\mathord{\equiv_M} = \mathord{\sim_N}$, where $N = \syn(M')$.
\end{thm}

\begin{proof}
  Lemma~\ref{lm:cong} shows that $\equiv_M$~is a congruence of~$N$
  that respects~$F$.  Since $\sim_N$ is the coarsest congruence of~$N$
  that respects~$F$ by~\cite[Theorem~II.6.10]{gecste84}, we obtain
  that $\mathord{\equiv_M} \subseteq \mathord{\sim_N}$.  The converse
  is simple to prove as states that are weakly equivalent
  in~$\syn(M')$ share exactly the same signs of life with the scaling
  factor~$1$.  Since the signs of life already indicate the transition
  weights, we immediately obtain that such weakly equivalent states
  in~$\syn(M')$ have corresponding transitions with equal transition
  weights in~$M'$, which proves that those states are also equivalent
  in~$M'$ with the scaling factor~$1$.  The latter statement can then
  be used to prove that they are also equivalent in~$M$ (with a
  scaling factor that is potentially different from~$1$).
\end{proof}

The currently fastest dwta minimization algorithm~\cite{mal08e} runs
in time~$\mathcal O \bigl(\abs M \cdot \abs Q \bigr)$.  Our approach,
which relies on pushing and is presented in
Algorithm~\ref{alg:Overall}, achieves the same run-time~$\mathcal
O \bigl(\abs M \log{\abs Q} \bigr)$ as the fastest minimization
algorithms for deterministic fta.

\begin{cor}[see Algorithm~\ref{alg:Overall}]
  \label{cor:Main}
  For every dwta~$M = (Q, \Sigma, \mu, F)$, we can compute an
  equivalent minimal dwta in time~$\mathcal O \bigl(\abs M \log {\abs
    Q} \bigr)$.
\end{cor}

\section{Testing equivalence}
\label{sec:equiv}
In this final section, we want to decide whether two given dwta are
equivalent.  To this end, let $M = (Q, \Sigma, \mu, F)$ and $M' = (Q',
\Sigma, \mu', F')$ be dwta.  The overall approach is presented in
Alg.~\ref{alg:Overall2}.  First, we compute a correspondence~$g
\colon Q \to Q'$ between states.  For every $q \in Q$, we compute a
tree~$t \in  T_\Sigma$, which is also called \emph{access tree
  for~$q$}, such that $h_\mu^{(1)}(t) = q$.  If no access tree exists,
then $q$~is not reachable and can be deleted.  A dwta, in which all
states are reachable, is called \emph{accessible}.  To avoid these
details, let us assume that $M$~and~$M'$ are accessible, which can
always be achieved in time~$\mathcal O \bigl(\abs M + \abs{M'}
\bigr)$.  In this case, we can compute an access tree~$a(q) \in
T_\Sigma$ for every state~$q \in Q$ in time~$\mathcal O \bigl(\abs M
\bigr)$ using standard breadth-first search, in which we unfold each
state (\ie, explore all transitions leading to it) at most once.  To
keep the representation efficient, we store the access trees in the
format~$\Sigma(Q)$, where each state~$q \in Q$ refers to its access
tree~$a(q)$.  To obtain the corresponding state~$g(q)$, we compute the
state of~$Q'$ that is reached when processing the access tree~$a(q)$.
Formally, $g(q) = h_{\mu'}^{(1)}(a(q))$ for every $q \in Q$.  This
computation can also be achieved in time~$\mathcal O \bigl(\abs M
\bigr)$ since we can reuse the results for the subtrees.
Consequently, we have that $h_\mu^{(1)}(a(q)) = q$ and
$h_{\mu'}^{(1)}(a(q)) = g(q)$ for every $q \in Q$.  Clearly, the
computation of the access trees~$a \colon Q \to T_\Sigma$ and the
correspondence~$g \colon Q \to Q'$ can be performed in time~$\mathcal
O \bigl(\abs M \bigr)$.  Next, we compute the coarsest congruences
$\sim_M$~and~$\sim_{M'}$ for $M$~and~$M'$ that respect $F$~and~$F'$,
respectively, and the signs of life for~$M$.

\begin{algorithm}[t]
  \begin{algorithmic}[2]
    \REQUIRE accessible dwta $M = (Q, \Sigma, \mu, F)$~and~$M' = (Q',
    \Sigma, \mu', F')$ 
    \ENSURE return `yes' if $M$~and~$M'$ are equivalent; `no' otherwise
    \smallskip\hrule\smallskip
    \STATE $g \gets \textsc{ComputeCorrespondence}(M, M')$
      \COMMENT{complexity: $\mathcal O \bigl(\abs M \bigr)$}%
    \STATE $\mathord{\sim_M} \gets
    \textsc{ComputeCoarsestCongruence}(M, Q \times Q)$
      \COMMENT{complexity: $\mathcal O \bigl(\abs M \log {\abs Q}
        \bigr)$}%
    \STATE $\mathord{\sim_{M'}} \gets 
    \textsc{ComputeCoarsestCongruence}(M', Q' \times Q')$
      \COMMENT{complexity: $\mathcal O \bigl(\abs{M'} \log {\abs{Q'}}
        \bigr)$}%
    \STATE $(L, \mathord{\sol}, \lambda) \gets \textsc{ComputeSoL}(M,
    \mathord{\sim_M})$
      \COMMENT{complexity: $\mathcal O \bigl(\abs M \bigr)$}%
    \IF{$g$ is not compatible with the congruences
      $\sim_M$~and~$\sim_{M'}$}
      \RETURN \textbf{no}
      \COMMENT{see Lemma~\protect{\ref{lm:compatible}}; complexity:
        $\mathcal O \bigl(\abs Q + \abs{Q'} \bigr)$}
    \ENDIF
    \FORALL{$q' \in Q'$}
      \STATE $\lambda'(q') = h_{\mu'}^{(2)}(c'[q'])$ with $c' =
      \ren_g(\sol(\overline g^{-1}([q']_{\sim_{M'}})))$
        \COMMENT{prepare pushing weights}%
    \ENDFOR
    \STATE $N \gets \textsc{Minimize}(\syn(\push_\lambda(M)),
    \mathord{\sim_M})$
      \COMMENT{complexity: $\mathcal O \bigl(\abs M \log {\abs Q}
        \bigr)$}%
    \STATE $N' \gets \textsc{Minimize}(\syn(\push_{\lambda'}(M')),
    \mathord{\sim_{M'}})$
      \COMMENT{complexity: $\mathcal O \bigl(\abs{M'} \log {\abs{Q'}}
        \bigr)$}%
    \RETURN $\textsc{Isomorphic?}(N, N')$
      \COMMENT{complexity: $\mathcal O \bigl(\abs N \bigr)$}%
  \end{algorithmic}
  \caption{{\protect{Overall structure of our equivalence test.}}}
  \label{alg:Overall2}
\end{algorithm}

\begin{lem}
  \label{lm:compatible}
  Let $M$~and~$M'$ be equivalent.  The correspondence~$g \colon Q \to
  Q'$ is compatible with the congruences $\sim_M$~and~$\sim_{M'}$;
  \ie, $g(q) \sim_{M'} g(p)$ if and only if $q \sim_M p$ for all $q, p
  \in Q$.  Moreover, for every reachable $q' \in Q'$ there exists $q
  \in Q$ such that $g(q) \in [q']_{\sim_{M'}}$.  Consequently, $g$ induces
  a bijection~$\overline g \colon (Q/\mathord{\sim_M}) \to
  (Q'/\mathord{\sim_{M'}})$ on the equivalence classes.
\end{lem}

\begin{proof}
  Let $q, p \in Q$, and let $t = a(q)$ and $u = a(p)$ be the
  corresponding access trees.  Then
  \begin{align*}
    &\phantom{{}\iff{}} q \sim_M p \\
    &\iff \{c \in C_\Sigma(Q) \mid h_\mu^{(1)}(c[q]) \in F \} = \{c \in
      C_\Sigma(Q) \mid h_\mu^{(1)}(c[p]) \in F\} \\
    &\iff \{c \in C_\Sigma \mid h_\mu^{(1)}(c[q]) \in F \} = \{c \in
      C_\Sigma \mid h_\mu^{(1)}(c[p]) \in F\} \tag{$\star$} \\
    &\iff \{c \in C_\Sigma \mid c[t] \in \supp(M) \} = \{c \in
      C_\Sigma \mid c[u] \in \supp(M)\} \tag{$\dagger$} \\
    &\iff \{c \in C_\Sigma \mid c[t] \in \supp(M') \} = \{c \in
      C_\Sigma \mid c[u] \in \supp(M')\} \tag{since $M = M'$} \\
    &\iff \{c \in C_\Sigma \mid h_{\mu'}^{(1)}(c[g(q)]) \in F' \} =
      \{c \in C_\Sigma \mid h_{\mu'}^{(1)}(c[g(p)]) \in F'\}
      \tag{$\dagger$} \\ 
    &\iff \{c \in C_\Sigma(Q) \mid h_{\mu'}^{(1)}(c[g(q)]) \in F' \} =
      \{c \in C_\Sigma(Q) \mid h_{\mu'}^{(1)}(c[g(p)]) \in F'\}
      \tag{$\star$} \\ 
    &\iff g(q) \sim_{M'} g(p) \enspace,
  \end{align*}
  where $(\star)$~follows from \cite[Lemma~4]{mal08b} and
  $(\dagger)$ follows from the easy fact that $h_\mu^{(1)}(c[q]) \in
  F$ if and only if $c[t] \in \supp(M)$ for all $q \in Q$ and $t \in
  T_\Sigma$ such that $h_\mu^{(1)}(t) = q$. 

  For the second statement, let $q' \in Q'$ be a reachable state, and
  let $t \in T_\Sigma$ be such that $h_{\mu'}^{(1)}(t) = q'$.
  Clearly, we have 
  \begin{align*}
    &\phantom{{}\stackrel{(\dagger)}={}} \{c \in C_\Sigma \mid
      h_{\mu'}^{(1)}(c[q']) \in F'\} \stackrel{(\dagger)}= \{c \in
      C_\Sigma \mid c[t] \in \supp(M')\} = \{c \in C_\Sigma \mid
      c[t] \in \supp(M)\} \\
    &\stackrel{(\dagger)}= \{c \in C_\Sigma \mid
      h_\mu^{(1)}(c[q]) \in F\} \stackrel{(\dagger)}= \{c \in C_\Sigma \mid
      c[a(q)] \in \supp(M)\} \\
    &\stackrel{\phantom{(\dagger)}}= \{c \in C_\Sigma \mid c[a(q)] \in
      \supp(M')\} \stackrel{(\dagger)}= \{c \in C_\Sigma \mid
      h_{\mu'}^{(1)}(c[g(q)]) \in F'\}
  \end{align*}
  where $q = h_\mu^{(1)}(t)$.  Consequently, using~$(\star)$ we obtain
  $q' \sim_{M'} g(q)$.
\end{proof}

We just demonstrated that for equivalent dwta the correspondence~$g$
always yields a bijection $\overline g \colon (Q/\mathord{\sim_M}) \to
(Q'/\mathord{\sim_{M'}})$.  We can test the compatibility in
time~$\mathcal O \bigl(\abs Q + \abs{Q'} \bigr)$.  Next we transfer
the signs of life via~$\overline g$ to the equivalence classes
of~$\sim_{M'}$ and calculate the corresponding pushing weights for all
states~$q' \in Q'$.  Since the signs of life can contain states
of~$Q$, we need to rename them using the correspondence~$g$, so we use
the function~$\ren_g \colon T_\Sigma(Q \cup \{\SBox\}) \to T_\Sigma(Q'
\cup \{\SBox\})$, which is defined by $\ren_g(\SBox) = \SBox$, 
$\ren_g(q) = g(q)$ for all $q \in Q$, and $\ren_g(\sigma(\seq t1k)) =
\sigma(\ren_g(t_1), \dotsc, \ren_g(t_k))$ for all $\sigma \in
\Sigma_k$ and trees $\seq t1k \in T_\Sigma(Q \cup \{\SBox\})$.   We note
that $\ren_g(c) \in C_\Sigma(Q')$ for all $c \in C_\Sigma(Q)$. 

Using this approach corresponding equivalence classes receive the same
sign of life (modulo the renaming~$\ren_g$ of the states).  We then
minimize $M$~and~$M'$ using the method of Section~\ref{sec:Min} (\ie,
we perform pushing followed by unweighted minimization).  Finally, we
test the obtained deterministic fta for isomorphism.  

\begin{lem}
  \label{lm:equivtest}
  We use the symbols of Algorithm~\ref{alg:Overall2}.  Given a
  compatible correspondence~$g$, the dwta $M$~and~$M'$ are equivalent
  if and only if the deterministic unweighted fta
  $\syn(\push_\lambda(M))$ and $\syn(\push_{\lambda'}(M'))$ are
  equivalent.
\end{lem}

\begin{proof}
  Clearly, if the deterministic fta
  $\syn(\push_\lambda(M))$~and~$\syn(\push_{\lambda'}(M'))$ are
  equivalent, then also $\push_\lambda(M)$~and~$\push_{\lambda'}(M')$
  are equivalent since the weights are annotated on the symbols of the
  former devices.  Moreover, since pushing preserves the semantics
  (see Proposition~\ref{prop:pushEquiv}), also the dwta $M$~and~$M'$
  are equivalent, which concludes one direction.  For the other
  direction, let $M$~and~$M'$ be equivalent.  Then also
  $\push_\lambda(M)$~and~$\push_{\lambda'}(M')$ are equivalent due to
  Proposition~\ref{prop:pushEquiv}.  An easy adaptation of the proof
  (of the equality~\eqref{eq:mn} of the transition weights) of
  Lemma~\ref{lm:cong} can be used to show that the transition weights
  of corresponding transitions are equal and hence
  $\syn(\push_\lambda(M))$~and~$\syn(\push_{\lambda'}(M'))$ are
  equivalent. 
\end{proof}

Lemma~\ref{lm:equivtest} proves the correctness of
Algorithm~\ref{alg:Overall2} because the minimal deterministic fta for
a given tree language is unique (up to
isomorphism)~\cite[Theorem~2.11.12]{gecste84}.  The run-time of our
algorithm should be compared to the previously (asymptotically)
fastest equivalence test for dwta of~\cite{drehogmal09b}, which runs
in time~$\mathcal O \bigl(\abs{M} \cdot \abs{M'} \bigr)$.

\begin{thm}
  \label{thm:Main2}
  We can test equivalence of $M$~and~$M'$ in time~$\mathcal O
  \bigl((\abs M + \abs{M'}) \log {(\abs Q + \abs{Q'})} \bigr)$.
\end{thm}

\section*{Acknowledgments}
The authors gratefully acknowledge the insight and suggestions
provided by the reviewers of the conference and the current version.

\bibliographystyle{alpha}
\bibliography{extra}

\end{document}